\documentclass[11pt]{article}

\usepackage{booktabs} 

\usepackage{amsmath,amsthm,amsfonts,amssymb,mathrsfs,bm,bbm}
\usepackage{xspace}
\usepackage{authblk}
\usepackage[margin=1in]{geometry}
\usepackage{comment}
\usepackage[T1]{fontenc}    
\usepackage{hyperref}       
\usepackage{url}            
\usepackage{nicefrac}       
\usepackage{microtype}      
\usepackage{cleveref}
\usepackage{txfonts}
\usepackage{xcolor}
\usepackage{libertine}

\usepackage[suppress]{color-edits}    
\addauthor{rdk}{blue}
\addauthor{ma}{red}
\addauthor{od}{purple}

\newtheorem{theorem}{Theorem}
\newtheorem{proposition}{Proposition}
\newtheorem{lemma}{Lemma}
\theoremstyle{definition}
\newtheorem{definition}{Definition}
\newtheorem{example}{Example}
\newtheorem{remark}{Remark}

\newcommand{\Exp}[1]{\mathbb{E}\left[#1\right]}
\newcommand{\expect}{\mathbb{E}}
\renewcommand{\Pr}{\text{Pr}}
\newcommand{\A}{\mathcal{A}}
\newcommand{\Rea}{\mathbb{R}}
\newcommand{\Rev}{\text{Rev}}
\newcommand{\Mat}{\mathcal{M}}
\newcommand{\Ind}{\mathcal{I}}
\newcommand{\MyerOPT}{\text{MyerOPT}}

\newcommand{\ones}{{\mathbbm{1}}}
\renewcommand{\vec}[1]{\mathbf{#1}}

\newcommand{\argmax}{\text{argmax}}

\title{Revenue Monotonicity Under Misspecified Bidders}
\date{}
\author{Makis Arsenis}
\author{Odysseas Drosis}
\author{Robert Kleinberg}
\affil{Cornell University, Ithaca, NY 14853, USA.\footnote{
Email: \texttt{\{marsenis,rdk\}@cs.cornell.edu, odysseas-11@hotmail.com}.}}

\begin{document}

\maketitle

\begin{abstract}
We investigate revenue guarantees for auction mechanisms in
a model where a distribution is specified for each bidder,
but only some of the distributions are correct.  The subset
of bidders whose distribution is correctly specified
(henceforth, the ``green bidders'') is unknown to the
auctioneer.  The question we address is whether the
auctioneer can run a mechanism that is guaranteed to obtain
at least as much revenue, in expectation, as would be
obtained by running an optimal mechanism on the green
bidders only. For single-parameter feasibility environments,
we find that the answer depends on the feasibility
constraint. For matroid environments, running the optimal
mechanism using all the specified distributions (including
the incorrect ones) guarantees at least as much revenue in
expectation as running the optimal mechanism on the green
bidders. For any feasibility constraint that is not a
matroid, there exists a way of setting the specified
distributions and the true distributions such that the
opposite conclusion holds.
\end{abstract}

\section{Introduction}

In a seminal paper nearly forty years ago~\cite{myerson81},
Roger Myerson derived a beautifully precise characterization
of optimal (i.e., revenue maximizing) mechanisms for Bayesian
single-parameter environments. One way this result has been
critiqued over the years is by noting that auctioneers may
have incorrect beliefs about bidders' values; if so, the
mechanism recommended by the theory will actually be
suboptimal.

In this paper we evaluate this critique by examining revenue
guarantees for optimal mechanisms when a subset of bidders'
value distributions are misspecified, but the auctioneer
doesn't know which of the distributions are incorrect. Our
model is inspired by the literature on {\em semi-random
adversaries} in the theoretical computer science literature,
particularly the work of Bradac et al.~\cite{BGSZ19} on
robust algorithms for the secretary problem.  In the model
we investigate here, the auctioneer is given (not
necessarily identical) distributions for each of $n$
bidders.  An unknown subset of the bidders, called the {\em
green bidders}, draw their values independently at random
from these distributions.  The other bidders, called the
{\em red bidders}, draw their values from distributions
other than the given ones.

The question we ask in this paper is, ``When can one
guarantee that the expected revenue of the optimal mechanism
for the given distributions is at least as great as the
expected revenue that would be obtained by excluding the red
bidders and running an optimal mechanism on the green subset
of bidders?'' In other words, can the presence of bidders
with misspecified distributions in a market be worse (for
the auctioneer's expected revenue) than if those bidders
were absent?  Or does the increased competition from
incorporating the red bidders always offset the revenue loss
due to ascribing the wrong distribution to them?

We give a precise answer to this question, for
single-parameter feasibility environments. We show that the
answer depends on the structure of the feasibility
constraint that defines which sets of bidders may win the
auction. For matroid feasibility constraints, the revenue of
the optimal mechanism is always greater than or equal to the
revenue obtained by running the optimal mechanism on the set
of green bidders. For any feasibility constraint that is not
a matroid, the opposite holds true: there is a way of
setting the specified distribution and the true
distributions such that the revenue of the optimal mechanism
for the specified distributions, when bids are drawn from
the true distributions, is {\em strictly less} than the
revenue of the optimal mechanism on the green bidders only.

The economic intuition behind this result is fairly easy to
explain. The matroid property guarantees that the winning
red bidders in the auction can be put in one-to-one
correspondence with losing green bidders who would have won
in the absence of their red competitors, in such a way that
the revenue collected from each winning red bidder offsets
the lost revenue from the corresponding green bidder whom he
or she displaces.  When the feasibility constraint is not a
matroid, this one-to-one correspondence does not always
exist; a single green bidder might be displaced by two or
more red bidders each of whom pays almost nothing. The
optimal mechanism allows this to happen at some bid
profiles, because the low revenue received on such bid
profiles is compensated by the high expected revenue that
would be received if the red bidders had sampled values from
elsewhere in their distributions. However, since the red
bidders' distributions are misspecified, the anticipated
revenue from these more favorable bid profiles may never
materialize.

Our result can be interpreted as a type of revenue
monotonicity statement for optimal mechanisms in
single-parameter matroid environments. However it does not
follow from other known results on revenue monotonicity, and
it is illuminating to draw some points of distinction
between our result and earlier ones. Let us begin by
distinguishing {\em pointwise} and {\em setwise} revenue
monotonicity results: the former concern how the revenue
earned on individual bid profiles varies as the bids are
increased, the latter concern how (expected) revenue varies
as the set of bidders is enlarged.
\begin{itemize}
  \item VCG mechanisms are neither pointwise
  nor setwise revenue monotone in general,
  but in single-parameter matroid feasibility
  environments, VCG revenue satisfies
  both pointwise and setwise monotonicity. In fact,
  Dughmi, Roughgarden, and Soundararajan~\cite{drs09}
  observed that VCG revenue obeys setwise
  monotonicity {\em if and only if} the
  feasibility constraint is a matroid.
  The proof of this result in~\cite{drs09}
  rests on a slightly erroneous characterization
  of matroids, and one (small) contribution of
  our work is to correct this minor error by
  substituting a valid characterization
  of matroids, namely \Cref{lem:nonmatroid}
  below.
  \item Myerson's optimal mechanism is not
  pointwise revenue monotone, even for single-item
  auctions. For example, consider using Myerson's optimal
  mechanism to sell a single item
  to Alice whose value is uniformly distributed
  in $[0,4]$ and Bob whose value is uniformly
  distributed in $[0,8]$. When Alice
  bids 0 and Bob bids 5, Bob wins and pays 4.
  If Alice increases her bid to 4, she wins but pays
  only 3.
  \item However, Myerson's optimal mechanism is
  {\em always} setwise revenue monotone in
  single-parameter environments with downward-closed
  feasibility constraints, regardless of
  whether the feasibility constraint is a matroid.
  This is because the mechanism's expected revenue
  is equal to the expectation of the maximum, over
  all feasible sets of winners, of the winners'
  combined ironed virtual value. Enlarging the
  set of bidders only enlarges the collection
  of sets over which this maximization is performed,
  hence it cannot decrease the expectation of the
  maximum.
\end{itemize}
Our main result is analogous to the setwise revenue
monotonicity of Myerson revenue, except that we are
considering monotonicity with respect to the operation of
enlarging the set of bidders {\em by adding bidders whose
value distributions are potentially misspecified.} We show
that the behavior of Myerson revenue with respect to this
stricter notion of setwise revenue monotonicity holds under
matroid feasibility constraints {\em but not under any other
feasibility constraints}, in contrast to the traditional
setwise revenue monotonicity that is satisfied by Myerson
mechanisms under arbitrarily downward-closed constraints.

\subsection{Related Work}

{\em Semi-random models} are a class of models
studied in the theoretical computer science literature
in which the input data is partly generated by random
sampling, and partly by a worst-case adversary.
Initially studied in the setting of
graph coloring~\cite{BlumSpencer} and
graph partitioning~\cite{FeigeKilian,mmv12},
the study of semi-random models has since been
broadened to
statistical estimation~\cite{diakonikolas2019robust,lai2016agnostic},
multi-armed bandits~\cite{lykouris2018stochastic},
and secretary problems~\cite{BGSZ19}.
Our work extends semi-random models into the
realm of Bayesian mechanism design. In particular,
our model of green and red bidders resembles in a sense that of
Bradac el al.~\cite{BGSZ19} for the secretary problem which
served as inspiration for this work. In both settings,
green players/elements behave randomly and independently
while red players/elements behave adversarially. In the
secretary model of~\cite{BGSZ19},
red elements can choose arbitrary arrival times
while green elements' arrival
times are i.i.d.~uniform in $[0,1]$ and
independent of the red arrival times.
Similarly, in our setting red bidders
can set their bids arbitrarily whereas
green bidders sample their bids from known
distributions, independently
of the red bidders and one another.

Our work can be seen as part of a general framework of {\em
robust mechanism design}, a research direction inspired by
Wilson~\cite{wilson1987}, who famously wrote,
\begin{quotation} Game  theory  has  a  great  advantage  in
explicitly analyzing  the  consequences  of  trading  rules
that presumably are really common knowledge; it is deficient
to the extent it assumes other features to be common
knowledge, such as one agent’s probability assessment about
another’s  preferences  or  information.  I  foresee the
progress of game theory as depending on successive
reductions in the base of common knowledge required to
conduct useful analyses of practical problems. Only by
repeated weakening of common knowledge assumptions will the
theory approximate reality.  \end{quotation} This {\em
Wilson doctrine} has been used to justify more robust
solution concepts such as dominant strategy and ex post
implementation. The question of when these stronger solution
concepts are required in order to ensure robustness was
explored in a research program initiated by Bergemann and
Morris~\cite{bergemann05} and surveyed
in~\cite{bm-survey-robust}.  Robustness and the Wilson
doctrine have also been used to justify
prior-free~\cite{Goldberg01} and
prior-independent~\cite{hr09} mechanisms as well as
mechanisms that learn from
samples~\cite{bubeck2019multi,cole-roughgarden,singlesample,huang2018making,morgenstern-roughgarden}.
A different approach to robust mechanism design assumes
that, rather than being given the bid distributions, the
designer is given constraints on the set of potential bid
distributions and aims to optimize a minimax objective on
the expected revenue. For example Azar and
Micali~\cite{AzarMicali13} assume the seller knows only the
mean and variance of each bidder's distribution, Carrasco et
al.~\cite{carrasc18} generalize this to sellers that know
the first $N$ moments of each bidder's distribution, Azar et
al.~\cite{AMDW13} consider sellers that know the median or
other quantiles of the distributions, Bergemann and
Schlag~\cite{bergemann2011robust} assume the seller is given
distributions that are known to lie in a small neighborhood
of the true distributions, and
Carroll~\cite{carroll2017robustness} introduced a model in
which bids are correlated but the seller only knows each
bidder's marginal distribution
(see~\cite{gravin2018separation,bei2019correlation} for
further work in this correlation-robust model).

Another related subject is that of {\em revenue
monotonicity} of mechanisms --- regardless of the existence
of adversarial bidders.  Dughmi et al.~\cite{drs09} prove a
result very close in spirit to ours. They consider the VCG
mechanism in a single-parameter downward-closed environment
and prove that it is revenue monotone if and only if the
environment is a matroid akin to our Theorems
\ref{thm:rmmb-matroid} and \ref{thm:rmmb-non-matroid}.

Rastegari et al.~\cite{RCLB07} study revenue monotonicity
properties of mechanisms (including VCG) for Combinatorial
Auctions. Under some reasonable assumptions, they prove that
no mechanism can be revenue monotone when bidders have
single-minded valuations.

\section{Preliminaries}

\subsection{Matroids}
\label{sec.matroids}

Given a finite ground set $E$ and a collection $\Ind
\subseteq 2^E$ of subsets of $E$ such that $\emptyset \in
\Ind$, we call $\Mat = (E, \Ind)$ a {\em set system}. $\Mat$
is a {\em downward-closed} set system if $\Ind$ satisfies
the following property:

\begin{enumerate}
\item[(I1)] {\bf (downward-closed axiom)} If $B \in \Ind$
and $A \subseteq B$ then $A \in \Ind$.
\end{enumerate}

\noindent
Furthermore, $\Mat$ is called a {\em matroid} if it satisfies
both (I1) and (I2):

\begin{enumerate}
\item[(I2)] {\bf (exchange axiom)} If $A, B \in \Ind$ and
$|A| > |B|$ then there exists $x \in A \backslash B$ such
that $B + x \in \Ind$\footnote{We use the shorthand $B + x$
(resp.~$B - x$) to mean $B \cup \{x\}$ (resp.~$B \backslash
\{x\}$) throughout the paper.}.
\end{enumerate}

In the context of matroids, sets in (resp.~not in) $\Ind$
are called {\em independent} (resp.~{\em dependent}). An
(inclusion-wise) maximal independent set is called a {\em
basis}. A fundamental consequence of axioms (I1), (I2) is
that all bases of a matroid have equal cardinality and this
common quantity is called the {\em rank} of the matroid. A
{\em circuit} is a minimal dependent set. The set of all
circuits of a matroid will be denoted by $\mathcal{C}$. The
following is a standard property of $\mathcal{C}$.

\begin{proposition}[{\cite[Proposition 1.4.11]{oxley}}]
\label{prop:circuits}
For any $\mathcal{C}$ which is the circuit set of a matroid
$\Mat$, let $C_1, C_2 \in \mathcal{C}, e \in C_1 \cap C_2$
and $f \in C_1 \backslash C_2$. Then there exists $C_3 \in
\mathcal{C}$ such that $f \in C_3 \subseteq (C_1 \cup C_2) -
e$.
\end{proposition}

For any set system $\Mat = (E, \Ind)$ and any given $S
\subseteq E$, define $\Ind_{\mid S} = \Ind \cap 2^S$ and
call $\Mat_{\mid S} = (S, \Ind_{\mid S})$ the {\em
restriction} of $\Mat$ on $S$. Notice that restrictions
maintain properties (I1), (I2) if they were satisfied
already in $\Mat$.

In what follows, we provide some examples of common matroids.
The reader is invited to check that they indeed satisfy
(I1), (I2). For a more in-depth study of matroid theory, we
point the reader to the classic text of Oxley \cite{oxley}.

\paragraph{Uniform matroids}
When $\Ind = \{ S \subseteq E : |S| \leq k \}$ for some
positive integer $k \leq |E|$, then $(E, \Ind)$ is called a
{\em uniform (rank $k$)} matroid.

\paragraph{Graphic matroids}
Given a graph $G = (V, E)$ (possibly containing parallel edges
and self-loops) let $\Ind$ include all subsets of edges
which do {\em not} form a cycle, i.e. the subgraph $G[S] =
(V, S)$ is a forest. Then $(E, \Ind)$ forms a matroid called {\em
graphic} matroid. \macomment{Oxley calls those matroids
cyclic and then defines graphic matroids to be any matroid
isomorphic to a cyclic matroid. I think this is not an
important distinction for our purposes .}
Graphic matroids capture many of the properties of general
matroids and notions like bases and circuits have their
graphic counterparts of spanning trees and cycles
respectively.

\paragraph{Transversal matroids}
Let $G = (A \cup B, E)$ be a simple, bipartite graph and
define $\Ind$ to include all subsets $S \subseteq A$ of
vertices for which the induced subgraph on $S \cup B$
contains a matching covering all vertices in $S$. Then $(A,
\Ind)$ is called a {\em transversal} matroid.

If $\Mat = (E, \Ind)$ is equipped with a weight function $w : E
\rightarrow \Rea^+$ it is called a {\em weighted} matroid.
The problem of finding an independent set of maximum sum of
weights\footnote{An equivalent formulation asks for a basis
of maximum total weight.}
is central to the study of matroids. A very simple greedy
algorithm is guaranteed to find the optimal solution and in
fact matroids are exactly the downward-closed systems for
which that greedy algorithm is always guaranteed to find the
optimal solution.

\paragraph{{\sc Greedy}}
Sort the elements of $E$ in non-increasing order of weights
$w(e_1) \geq w(e_2) \geq \ldots \geq w(e_n)$. Loop through
the elements in that order adding each element to the current
solution as long as the current solution remains an
independent set.

\begin{lemma}[{\cite[Lemma 1.8.3.]{oxley}}]
Let $\Mat = (E, \Ind)$ be a weighted downward-closed set
system. Then {\sc Greedy} is guaranteed to return an
independent set of maximum total weight for every weight
function $w : E \rightarrow \Rea^+$ if and only $\Mat$ is a
matroid.
\end{lemma}

In what follows we're going to assume without loss of
generality that the function $w$ is one-to-one, meaning that
no two element have the same weight. All proofs can be
adapted to work in the general case using any deterministic
tie breaking rule.

The following proposition provides a convenient way for
updating the solution to an optimization problem under
matroid constraints when new elements are added. A proof is
included in the Appendix.

\begin{proposition}
\label{prop:matroid-update}
Let $\Mat = (E, \Ind)$ be a weighted matroid with weight
function $w : E \rightarrow
\Rea^+$. Consider the max-weight independent set $I$ of
the restricted matroid $\Mat_{\mid E - x}$. Then the
max-weight independent set $I^*$ of
$\Mat$ can be obtained from $I$ as follows: if $(I + x)
\in \Ind$ then $I^* = I + x$, otherwise, $I^* = (I + x) -
y$ where $y$ is the minimum-weight element in the unique
circuit $C$ of $I + x$.
\end{proposition}

\subsection{Optimal Mechanism Design}

We study auctions modeled as a {\em Bayesian
single-parameter environment}, a standard mechanism design
setting in which a {\em seller} (or mechanism designer)
holds many identical copies of an item they want to sell. A
set of $n$ bidders (or players), numbered $1$ through $n$,
participate in the auction and each bidder $i$ has a
private, non-negative value $\varv_i \sim F_i$, sampled
(independently across bidders) from a distribution $F_i$
known to the seller.
Abusing notation, we'll use $F_i$ to also denote the
cumulative distribution function and $f_i$ to denote the
probability density function of the respective distribution.
The value of each bidder expresses
their valuation for receiving one item.  Let $V_i$ be the
support of distribution $F_i$ and define $V = V_1 \times
\ldots \times V_n$. For a vector $\vec{v} \in V$, we use the
standard notation $\vec v_{-i} = (\varv_1, \ldots,
\varv_{i-1}, \varv_{i+1}, \ldots, \varv_n)$ to express the
vector of valuations of all bidders {\em except} bidder $i$.
When the index set $[n]$ is partitioned into two sets
$A,B$ and we have vectors $\vec{v}_A \in \Rea^A, \, \vec{w}_B \in \Rea^B$,
we will abuse notation and let $(\vec{v}_A,\vec{w}_B)$ denote the
vector obtained by interleaving
$\vec{v}_A$ and $\vec{w}_B$, i.e.~$( \vec{v}_A,\vec{w}_B)$
is the vector $\vec{u} \in \Rea^n$
specified by
\[ u_i = \begin{cases} v_i & \mbox{if } i \in A \\
                       w_i & \mbox{if } i \in B. \end{cases} \]
Similarly, when $\vec{v} \in V$, $i \in [n]$, and $z \in \Rea$,
$(z,\vec v_{-i})$ will denote the vector obtained by replacing the
$i^{\mathrm{th}}$ component of $\vec{v}$ with $z$.

A {\em feasibility constraint} $I \subseteq 2^{[n]}$ defines all
subsets of bidders that can be simultaneously declared
winners of the auction. We will interchangeably denote
elements of $I$ both as subsets of $[n]$ and as vectors
in $\{0, 1\}^n$. Of special interest are feasibility
constraints which define the independent sets of a matroid. We will
sometimes use the phrase {\em matroid market} to indicate this
fact. Matroid markets model many real world applications.
For example when selling $k$ identical copies of an item, the market is
a uniform rank $k$ matroid. Another example is kidney
exchange markets which can be modeled as transversal
matroids (\cite{Roth05}).

In a {\em sealed-bid auction}, each bidder $i$ submits a
{\em bid} $b_i \in V_i$ simultaneously to the mechanism.
Formally, a {\em mechanism} $\mathcal{A}$ is a pair $(x, p)$ of an
allocation rule $x : V \rightarrow I$ accepting the bids and
choosing a feasible outcome and a {\em payment rule} $p : V
\rightarrow \Rea^n$ assigning each bidder a monetary payment
they need to make to the mechanism. We denote by $x_i(\vec
b)$ (or just $x_i$ when clear from the context) the $i$-th
component of the 0-1 vector $x(\vec b)$ and similarly for
$p$.  An allocation rule is called {\em monotone} if the
function $x_i(z, \vec b_{-i})$ is monotone non-decreasing in
$z$ for any vector $\vec b_{-i} \in V_{-i}$ and any bidder
$i$.

We assume bidders have {\em quasilinear utilities}
meaning that bidder's $i$ utility for winning the auction and
having to pay a price $p_i$ is $u_i = \varv_i - p_i$ and $0$ if
they do not win and pay nothing. Bidders are selfish agents aiming to
maximize their own utility.

A mechanism is called {\em truthful} if bidding $b_i =
\varv_i$ is a {\em dominant strategy} for each bidder, i.e.
no bidder can increase their utility by reporting $b_i \neq
\varv_i$ regardless the values and bids of the other
bidders. An allocation rule $x$ is called {\em
implementable} if there exists a payment rule $p$ such that
$(x, p)$ is truthful. Such mechanisms are well understood
and easy to reason about since we can predict how the
bidders are going to behave.  In what follows we focus our
attention only on truthful mechanisms and thus use the terms
value and bid interchangeably.

A well known result of Myerson (\cite{myerson81}) states that a given
allocation rule $x$ is implementable if and only if $x$ is
monotone. In case $x$ is monotone, Myerson gives an explicit
formula for the unique\footnote{Unique up to the normalizing
assumption that $p_i = 0$ whenever $b_i = 0$.} payment rule
such that $(x, p)$ is truthful.
In the single-parameter setting we're studying, the payment
rule can be informally described as follows: $p_i$ is equal
to the minimum $b_i$ that bidder $i$ has to report such that
they are included in the set of winners --- we'll refer to
such a $b_i$ as the {\em critical bid} of bidder $i$.

The mechanism designer, who is collecting all the payments,
commonly aims to maximize her {\em expected revenue} which
for a mechanism $\mathcal{A}$ is defined as
$\text{Rev}(\mathcal{A}) = \expect_{b_i \sim F_i} \left[
\sum_{i \in [n]} p_i \right]$.

\begin{lemma}[{\cite{myerson81}}]
\label{lem:myerson}
For any truthful mechanism $(x, p)$ and any bidder $i \in [n]$:

\[ \Exp{ p_i } = \Exp{
\phi_i(b_i) \cdot x_i(b_i, \vec{b_{-i}}) } \]

\noindent
where the expectations are taken over $b_1, \ldots, b_n \sim
F_1, \ldots, F_n$, the function $\phi_i(\cdot)$ is defined
as

\[ \phi_i(z) = z - \frac{1 - F_i(z)}{f_i(z)} \]

\noindent
and $\phi_i(b_i)$ is called the {\em virtual value} of bidder $i$.
\end{lemma}

The importance of this lemma is that it reduces the problem
of revenue maximization to that of virtual welfare
maximization. More specifically, consider a sequence of
distributions $F_1, \ldots, F_n$ which have the property that all $\phi_i$
are monotone non-decreasing (such distributions are called
{\em regular}). In this case, the allocation rule that
chooses a set of bidders with the maximum total virtual
value (subject to feasibility constraints) is monotone (a
consequence of the regularity condition) and thus
implementable. We'll frequently denote this
revenue-maximizing mechanism by MyerOPT.

More precisely, the MyerOPT mechanism works as follows:
\begin{itemize}
\item Collect bids $b_i$ from every bidder $i \in [n]$.
\item Compute $\phi_i(b_i)$ and discard all bidders whose
virtual valuation is negative.
\item Solve the optimization problem $S^* = \argmax_{ S \in I }
\sum_{i \in S} \phi_i(b_i)$.
\item Allocate the items to $S^*$ and charge each bidder $i
\in S^*$ their critical bid.
\end{itemize}

Handling non-regular distributions is possible using the
standard technique of {\em ironing}. Very briefly, it works
as follows. So far, we've been expressing $x, p$ and $\phi$
as a function of the random vector $\vec v$. It is
convenient to switch to the quantile space and express them
as a function of a vector $\vec q \in [0, 1]^n$ where for a
given sample $z$ from $F_i$ we let $q_i = \Pr_{b_i \sim F_i}
[ b_i \geq z ]$. Another way to think of this is, instead of
sampling values, we sample quantiles $q_i$ distributed
uniformly at random in the interval $[0, 1]$ which are then
transformed into values $v_i(q_i) =
F_i^{-1}(1-q_i)$\footnote{In general, $v_i(q_i) = \min \left
\{ v \mid F_i(v) \geq q_i \right\}$.}. Let $R_i(q_i) = q_i
\cdot v_i(q_i)$ and notice that $\phi_i(v_i(q_i)) = \left.
\tfrac{dR_i}{dq} \right|_{q = q_i}$. Now, since $v_i(\cdot)$
is a non-increasing function we have that $\phi_i(\cdot)$ is
monotone if and only if $R$ is concave.

Now, suppose that $F_i$ is such that $R_i$ is not concave.
One can consider the concave hull of $\overline{R_i}$ of
$R_i$ which replaces $R_i$ with a straight line in every
interval that $R_i$ was not following that concave hull. The
corresponding function $\overline{\phi_i}( \cdot ) =
\tfrac{d\overline R_i}{dq}$ is called {\em ironed virtual
value function}.

\begin{lemma}[{\cite[Theorem 3.18]{HartlineBook}}]
For any monotone allocation rule $x$ and any virtual value
function $\phi_i$ of bidder $i$, the expected virtual
welfare of $i$ is
upper bounded by their expected ironed virtual value
welfare.
\[
\Exp{ \phi_i(v_i(q_i)) \cdot x_i( v_i(q_i),
\vec v_{-i}(\vec q) ) }
\leq
\Exp{ \overline{\phi_i}(v_i(q_i)) \cdot x_i( v_i(q_i),
\vec v_{-i}(\vec q) ) }
\]
Furthermore, the inequality holds with equality if the
allocation rule $x$ is such that for all bidders $i$,
$x_i'(q) = 0$ whenever $\overline{R_i}(q) > R_i(q)$.
\end{lemma}

As a consequence, consider the monotone allocation rule
which allocates to a feasible set of maximum total ironed virtual
value. On the intervals where $\overline{R_i}(q) > R_i(q)$,
$\overline{R_i}$ is linear as part of the concave hull so
the ironed virtual value function, being a derivative of a
linear function, is a constant. Therefore, the allocation
rule is not affected when $q$ ranges in such an interval.

A crucial property of any (ironed) virtual value function
$\phi$ corresponding to a distribution $F$ is that
$z \geq \phi(z)$ for all $z$ in the support of $F$.
This is obvious for $\phi$ as defined in \Cref{lem:myerson}.
We claim it also holds for ironed virtual value functions:
if $z$ lies in an interval where $\overline{\phi} = \phi$ it
holds trivially. Otherwise, if $z \in [a, b]$ for some
interval where $\phi$ needed ironing (i.e.~$\overline{R}(q)
> R(q)$ in the quantile space), we have: $z \geq a \geq
\phi(a) = \overline{\phi}(a) = \overline{\phi}(z)$. We've
thus proven:

\begin{proposition}
\label{prop:phi-inequality}
Any (possibly non-regular) distribution $F$ having an
ironed virtual value function $\overline{\phi}$ satisfies:
\[ z \geq \overline{\phi}(z) \]
for any $z$ in the support of $F$.
\end{proposition}

\begin{remark}
For simplicity, in the remainder of the paper we'll use
$\phi$ and $\overline{\phi}$ interchangeably and we will refer to
$\phi$ as virtual value function. The reader should keep in
mind that if the associated distribution is non-regular,
then {\em ironed} virtual value functions should be used
instead.
\end{remark}

\section{Revenue Monotonicity on Matroid Markets}
\label{sec:matroid-monotonicity}

We extend the standard single-parameter environment to allow
for bidders with misspecified distributions. Formally, the
$n$ bidders are partitioned into sets
\rdkdelete{$G \cup R = \{1,
\ldots, k\} \cup \{k + 1, \ldots, n\}$}
$G$ and $R$; the former are
called {\em green} and the latter {\em red}.
The color of
each bidder (green or red) is not revealed to the mechanism designer at
any point. Green bidders sample their values from their
respective distribution $F_i$ but red bidders are
sampling $v_i \sim F_i'$ for some $\{F_i'\}_{i \in R}$
which are completely unknown to the mechanism designer
and can be adversarially chosen.

In this section we are interested in studying the
behavior of Myerson's optimal mechanism when designed under
the (wrong) assumption that $v_i \sim F_i$ for all $i \in
[n]$. Specifically, we ask the question of whether the
existence of the red bidders could harm the expected revenue
of the seller compared to the case where the seller was able
to identify and exclude the red bidders, thus designing the
optimal mechanism for the green bidders alone. The following
definition makes this notion of revenue monotonicity more
precise.

\begin{definition}[RMMB]
\label{def:rmmb}
Consider a single-parameter, downward-closed market $\Mat =
(E, \Ind)$ of $|E| = n$ bidders.
A mechanism $\A$ is {\em Revenue Monotone under
Misspecified Bidders (RMMB)} if for any \madelete{regular}
distributions $F_1, \ldots, F_n$, any number $1 \leq k \leq n$ of
green bidders and any fixed misspecified bids $\vec{b}'_R \in \Rea^R$
of the red bidders:

\begin{equation}
\label{ineq:rmmb}
\Exp{ \Rev(
\A(b_G, b_R) ) } \geq
\Exp{ \Rev( \A(b_G)) }
\end{equation}

\noindent
where both expectation are taken over $b_G \sim
\prod_{i \in G} F_i$.
\end{definition}

An alternative definition of the revenue monotonicity property
allows red bidders to have stochastic valuations drawn from
distributions $F_i' \neq F_i$ instead of fixed bids. We note
that the two definitions are equivalent: if $\mathcal{A}$ is
RMMB according to \Cref{def:rmmb} then inequality
(\ref{ineq:rmmb}) holds pointwise for any fixed misspecified
bids and thus would also hold in expectation.  For the other
direction, if inequality (\ref{ineq:rmmb}) holds in
expectation over the red bids, regardless of the choice
of distributions $\{F_i' \mid i \in R\}$ then we may specialize
to the case when each $F_i'$ is a point-mass distribution with a
single support point $b_i$ for each $i \in R$, and then
\Cref{def:rmmb} follows.

In what follows we assume bidders always submit bids that
fall within the support of their respective distribution.
Green bidders obviously follow that rule and red bidders
should do as well, otherwise the mechanism could recognize
they are red and just ignore them.

Consider first the simpler case of selling a single item.
This corresponds to a uniform rank $1$ matroid market.
Intuitively when the item is allocated to a green bidder,
the existence of the red bidders is not problematic and in
fact could help increase the critical bid and thus the
payment of the winner.  On the other hand, when a red bidder
wins one has to prove that they are not charged too little
and thus risk bringing the expected revenue down.

Let $m = \max( \max_{i \in G} \phi_i(b_i), 0 )$ be the
random variable denoting the highest non-negative virtual
value in the set of green bidders. Let also $X$ be an
indicator a random variable which is 1
if the winner belongs to $G$ and $Y$
denote an indicator random variable which is 1 if the winner
belongs to $R$. For the mechanism MyerOPT have:
\begin{align}
\label{sec23:eq1}
\Exp{ \text{revenue from green bidders} } &= \Exp{m \cdot X}\\
\label{sec23:eq2}
\Exp{ \text{revenue from red bidders} } &\geq \Exp{m \cdot Y}
\end{align}

\noindent
where (\ref{sec23:eq1}) follows from Myerson's lemma and
(\ref{sec23:eq2}) follows from the observation that the
winner of the optimal auction never pays less than the
second-highest virtual value. To see why the latter holds,
let $\phi_s$ be the second highest
virtual value, $r$ be the red winner and $g$ is the green player
with the highest virtual value.
The critical bid of the red winner is at least
$\phi_r^{-1} (\phi_s)
\geq \phi_r^{-1} (\phi_g(b_g))
\geq \phi_g(b_g)$ where we applied the fact that $x
\geq \phi(x)$ to the virtual value function
$\phi = \phi_r$ and the value $x = \phi_r^{-1} (\phi_g(b_g))$.

Summing (\ref{sec23:eq1}) and (\ref{sec23:eq2}) and
using the fact that $X + Y = 1$ whenever $m > 0$, we find:

\begin{align*}
\Exp{\text{Revenue from all bidders } 1, \dots, n}
    &\geq \Exp{m \cdot X} + \Exp{m \cdot Y}\\
    &= \Exp{m}\\
    &= \Exp{\text{revenue of MyerOPT on } G}
\end{align*}

We therefore concluded that Myerson's optimal mechanism is
RMMB in the single-item case. We are now ready to
generalize the above idea to any matroid market.

\begin{theorem}
\label{thm:rmmb-matroid}
Let $\Mat = (E, \Ind)$ be any matroid market. Then
$\MyerOPT$ in $\Mat$ is RMMB
\end{theorem}

\begin{proof}
Call $G$ the set of green bidders and $R$
the set of red bidders.
Let $(x,p)$ denote the allocation and payment
rules for the mechanism $\MyerOPT$ that runs
Myerson's optimal mechanism on all $n$ bidders,
using the given distribution of each.
Let $(x',p')$ denote the allocation and payment
rules for the mechanism $\MyerOPT_G$ that runs
Myerson's optimal mechanism in the bidder set $G$ only.
For a set $S
\subseteq [n]$, let $T_S$ be the random variable
denoting the independent subset of $S$ that maximizes
the sum of ironed virtual values. In other words,
$T_S$ is the set of winners chosen by Myerson's
optimal mechanism on bidder set $S$.

By Myerson's Lemma, the revenue of
$\MyerOPT_G$ satisfies:

\begin{equation}
\label{eq0}
\Exp{\sum_{i \in G} p'_i(\vec b)} =
\Exp{ \sum_{i \in G} x'_i(\vec b) \cdot \phi_i(b_i) }
\end{equation}

By linearity
of expectation, we can
break up the expected revenue of
$\MyerOPT$ into two terms as follows:

\begin{equation}
\label{eq1}
\Exp{ \sum_{i \in [n]} p_i(\vec b) } = \Exp{ \sum_{i \in G} p_i(\vec b) } +
\Exp{ \sum_{i \in R} p_i(\vec b) }
\end{equation}

The first term on the right side of~\eqref{eq1}
expresses the revenue original from the green bidders.
Using Myerson's Lemma, we can equate this revenue
with the expectation of the
green winners' combined virtual value:

\begin{equation}
\label{eq2}
\Exp{ \sum_{i \in G} p_i(\vec b) } =
\Exp{ \sum_{i \in G} x_i(\vec b) \cdot \phi_i(b_i) } .
\end{equation}

To express the revenue coming from the red bidders in terms
of virtual valuations, we provide the argument that follows.
One way to derive $T_{G+R}$ from $T_G$ is to start with
$T_G$ and
sequentially add elements of $T_{G+R} \cap R$ in arbitrary
order while removing at each step the least weight element in the
circuit that potentially forms (repeated application of
\Cref{prop:matroid-update}). Let $e$ be the
$i$-th red element we're adding. If no circuit forms after
the addition, then
$e$ pays the smallest value in its support which is a
non-negative quantity. Otherwise, let $C$ be the
unique circuit that forms after that addition. Let $f$ be the
minimum weight element in $C$ and let $b_f$ be the associated
bid made by player $f$. Notice that $f$ must be green;
by assumption, every red element we're adding is part of the
eventual optimal solution so it cannot be removed at any
stage of this process.
The price charged to $e$ is their critical bid which we
claim is at least $\phi_e^{-1}(\phi_f(b_f))$. The reason is
that $e$ is part of circuit $C$ and $f$ is the min-weight
element of that circuit. The min-weight element of a circuit
is never in the max-weight independent set\footnote{This is a
consequence of the optimality of the {\sc Greedy} algorithm
since the min-weight element of a circuit is the last to be
considered among the elements of the circuit and its
inclusion will violate independence.} so if bidder $e$ bids
any value $v$ such that $\phi_e(v) < \phi_f(b_f)$ they will
certainly {\em not} be included in the set of winners, $T_{G+R}$.
By \Cref{prop:phi-inequality} it follows that $\phi_e^{-1}(
\phi_f(b_f) ) \geq \phi_f(b_f)$ thus $p_e(\vec b) \geq
\phi_f(b_f)$.

The above reasoning allows us ``charge'' each red bidder's
payment to a green player's virtual value in $T_G \setminus
T_{G+R}$:

\begin{align} \nonumber
\Exp{ \sum_{i \in R} p_i(\vec b) }
& \geq \Exp{ \sum_{i \in T_G \setminus T_{G+R}} \phi_i(b_i)} \\
& = \Exp{ \sum_{i \in G} (x'_i(\vec b) - x_i(\vec b)) \cdot \phi_i(b_i) }
\label{eq3}
\end{align}
The second line is justified by observing that for $i \in G$,
$x'_i(\vec b) = x_i(\vec b)$ unless $i \in T_G \setminus T_{G+R}$,
in which case $x'_i(\vec b) - x_i(\vec b) = 1$.

Combining Equations (\ref{eq0})-(\ref{eq3}) we get:

\begin{align*}
\Exp{ \sum_{i \in [n]} p_i(\vec b) } & \geq
    \Exp{ \sum_{ i \in G  } x_i(\vec b) \cdot \phi_i (b_i) }
    + \Exp{ \sum_{i \in G} (x'_i(\vec b) - x_i(\vec b)) \cdot \phi_i(b_i)} \\
    & = \Exp{ \sum_{i \in G } x'_i(\vec b) \cdot \phi_i (b_i) }
    = \Exp{ \sum_{i \in G} p'_i(\vec b) }
\end{align*}
In other words, the expected revenue of $\MyerOPT$
is greater than or equal to that of $\MyerOPT_G$.

\end{proof}

\section{General Downward-Closed Markets}

When the market is not a matroid, the existence of red
bidders can do a lot of damage to the revenue of the
mechanism as shown in the following simple example.

\begin{example}
\label{ex:nonmatroid}
Consider a 3-element downward-closed set system on
$E = \{a, b, c\}$ with maximal feasible sets: $\{a, b\}$ and $\{c\}$.
Let $c$ be a green bidder with a deterministic value of $1$
and $a, b$ be red bidders each with a specified value distribution
given by the following cumulative distribution function $F(x) = 1 - (1
+ x)^{1-N}$ for some parameter $N$.
Note that the associated virtual value function is:

\[
  \phi(x) = x - \frac{1-F(x)}{f(x)}
    = x - \frac{(1+x)^{1-N}}{(N-1) (1+x)^{-N}}
    = x - \tfrac{1+x}{N-1} = \left( 1 - \tfrac{1}{N-1} \right) x
    - \tfrac{1}{N-1} .
\]

For this virtual value function we have
$\phi^{-1}(0) = \frac{1}{N-2}$,
$\phi^{-1}(1) = \frac{N}{N-2}$.

Consider the revenue of Myerson's mechanism when
the red bidders, instead of following their
specified distribution, they each bid $\phi^{-1}(1)$ --- and
the green bidder bids $1$, the only support point of their
distribution. The set $\{a, b\}$ wins over $\{c\}$ since the
former sums to a total virtual value of $2$ over the latter's
virtual value $1$ so bidders $a, b$ pay their critical bid.

To compute that, notice that each of the bidders $a, b$
could unilaterally decrease their bid to any $\varepsilon >
\tfrac{1}{N-2}$ and they would still win the auction since
the set $\{a, b\}$ would still have a total virtual value greater
than $1$. Therefore, each of $a, b$ pays
$\tfrac{1}{N-2}$ for a total revenue of $\tfrac{2}{N-2}$.

On the other hand, the same mechanism when run on the set
$\{c\}$ of only the green bidder, always allocates an item
to $c$ and collects a total revenue of $1$.

Letting $N \rightarrow \infty$ we see that the former
revenue tends to zero while the latter remains $1$,
violating the revenue monotonicity property of
\Cref{def:rmmb} by an unbounded multiplicative factor.
\end{example}

To generalize the above idea to any non-matroid set system
we need the following lemma.

\begin{lemma}
\label{lem:nonmatroid}
A downward-closed set system $\mathcal{S} = (E, \Ind)$ is {\em not}
a matroid if and only if there exist $I, J \in \Ind$ with the following
properties:
\begin{enumerate}
    \item \label{lem:nonmatroid:propK}
        For every $K \in \Ind|_{I \cup J}$, if $|K| \geq |I|$
        then $K \supseteq I \backslash J$.
    \item \label{lem:nonmatroid:propNonEmpty}
        $|J \backslash I| \geq 1$.
    \item \label{lem:nonmatroid:propMaximum}
        $I$ is a maximum cardinality element of $\Ind|_{I \cup J}$.
\end{enumerate}
\end{lemma}

\begin{proof}
For the forward direction, suppose $\mathcal{S}$ is {\em
not} a matroid and let $V$ be a minimum-cardinality subset
of $E$ that is not a matroid.  Since $\Ind|_V$ is
downward-closed and non-empty, it must violate the exchange
axiom. Hence, there exist sets $I,J \in \Ind|_V$ such that
$|I| > |J|$ but $J+x \not\in \Ind$ for all $x \in I
\backslash J$. Note
that $V = I \cup J$, since otherwise $I \cup J$ is a
strictly smaller subset of $E$ satisfying the property that
$(I \cup J, \Ind|_{I \cup J})$ is not a matroid.

Observe that $J$ is a maximal element of $\Ind_V$. The
reason is that $V = I \cup J$, so every element of $V
\backslash J$ belongs to $I$. By our assumption on the pair
$I,J$, there is no element $y \in I$ such that $J + y \in
\Ind|_V$. Since $\Ind|_V$ is downward-closed, it follows that
no strict superset of $J$ belongs to $\Ind|_V$.

We now proceed to prove that $I, J$ satisfy the required
properties of the lemma:

\begin{enumerate}
\item[(\ref{lem:nonmatroid:propK})]
Let $K \in \Ind_V$ with $|K| \geq |I|$.  It follows that
$|K| > |J|$, but $J$ is maximal in $\Ind|_V$, so $K$ and $J$
must violate the exchange axiom.  Thus, $\Ind|_{K \cup J}$ is
not a matroid.  By the minimality of $V$, this implies $K
\cup J = V$ hence $K \supseteq I \backslash J$.

\item[(\ref{lem:nonmatroid:propNonEmpty})]
If $J \backslash I = \emptyset$ then $J \subseteq I$ which
contradicts the fact that $I, J$ violate the exchange axiom.

\item[(\ref{lem:nonmatroid:propMaximum})]
Suppose there exists $I' \in \Ind|_V$ with $|I'| > |I|$, then
by property (\ref{lem:nonmatroid:propK}) we have $I'
\supseteq I \backslash J$.  Remove elements of $I \backslash
J$ from $I'$ one by one, in arbitrary order, until we reach
a set $K \in \Ind_V$ such that $|K| = |I|$. This is possible
because after the entire set $I \backslash J$ is removed
from $I'$, what remains is a subset of $J$, hence has
strictly fewer elements than $I$. The set $K$ thus
constructed has $|K| = |I|$ but $K \not\supseteq I
\backslash J$, violating property
(\ref{lem:nonmatroid:propK}).
\end{enumerate}

For the ``only if'' direction, supposing that $\mathcal{S}$
is a matroid, we must show that no $I,J \in \Ind$ satisfy
all three properties.  To this end, suppose $I$ and $J$
satisfy (\ref{lem:nonmatroid:propNonEmpty}) and
(\ref{lem:nonmatroid:propMaximum}). Since $\mathcal{S}|_{I
\cup J}$ is a matroid, there exists $K \supseteq J$ such
that $K \in \Ind|_{I \cup J}$ and $|K| = |I|$.  By property
(\ref{lem:nonmatroid:propNonEmpty}), we know that no
$|I|$-element superset of $J$ contains $I-J$ as a subset.
Therefore, the set $K$ violates property
(\ref{lem:nonmatroid:propK}).
\end{proof}

We are now ready to generalize \Cref{ex:nonmatroid} to every non-matroid set
system.

\begin{theorem}
\label{thm:rmmb-non-matroid}
For any $\Mat = (E, \Ind)$ which is {\em not} a matroid,
MyerOPT is {\em not} RMMB.
\end{theorem}
\begin{proof}
Consider a downward-closed $\Mat = (E, \Ind)$ which is {\em
not} a matroid.  We are going to show there exists a
partition of players into green and red sets and a choice of
valuation distributions and misspecified red bids such that the
RMMB property is violated.

Let $I, J \subseteq E$ be the subsets whose existence is guaranteed
by \Cref{lem:nonmatroid}. Define $G = J$ to be the set of green bidders,
$R = I \backslash J$ to be the set of red bidders. All other bidders are irrelevant
and can be assumed to be bidding zero.
Set the value of each green bidder to be deterministically
equal to 1. For each red bidder $r$,
the specified value distribution has
the same cumulative distribution function
$F(x) = 1 - (1+x)^{1-N}$ defined in \Cref{ex:nonmatroid}.

Now let's consider the expected revenue of Myerson's
mechanism when every bidder in $R$ bids
$\phi^{-1}(1)$.\footnote{Members of $G$ bid as well, but it
hardly matters, because their bid is always 1 --- the only
support point of their value distribution --- so the
auctioneer knows their value without their having to submit
a bid.} Every bidder's virtual value is 1, so the mechanism
will choose any set of winners with maximum cardinality
which, according to \Cref{lem:nonmatroid}, property
(\ref{lem:nonmatroid:propMaximum}), is $|I|$.
For example, the set of winners could be $I$\footnote{In
general, the mechanism might choose any set $W$ of winners
such that $|W| = |I|$. The way to handle this case is
similar to the one used in the proof of \Cref{thm:corrected}
in the Appendix.}.

A consequence of \Cref{lem:nonmatroid}, property
(\ref{lem:nonmatroid:propK}) is that
for
every red bidder $r$ there
is no set of bidders disjoint from $\{r\}$ with
combined  virtual value greater than $|I|-1$.
Thus each red bidder pays $\phi^{-1}(0)$.
Elements of $I \cap J$ correspond to green
bidders who win the auction and pay 1, because
a green bidder pays 1 whenever they win. There are $|I \cap J|$
such bidders. Thus, the Myerson revenue is
$|I \cap J| + \frac{1}{N-2} |I \backslash J|$. The optimal
auction on the green bidders alone charges each
of these bidders a price of 1, receiving revenue
$|J| = |I \cap J| + |J \backslash I|$.
This exceeds $|I \cap J| + \frac{1}{N-2} |I \backslash J|$
as long as
\begin{equation} \label{eq:N-2}
  (N - 2) \cdot |J\backslash I| > |I \backslash J| .
\end{equation}
This inequality is satisfied, for example,
when $N = |I \backslash J| + 3$, because $J \backslash I$ has
at least one element (\Cref{lem:nonmatroid}, property
(\ref{lem:nonmatroid:propNonEmpty})).
\end{proof}

\section{Open Questions}

The previous section concluded with a proof that for any
non-matroid system, the ratio $r =
\frac{\Exp{\text{Rev}(G)}}{\Exp{\text{Rev}(G \cup R)}}$ for
Myerson's optimal mechanism can be greater than $1$.  An
interesting question is whether that ratio can be made
arbitrarily large as in \Cref{ex:nonmatroid}. If the sets
$I, J$ in the above proof are such that $I \cap J =
\emptyset$, then the ratio can be made unbounded with the
same construction. We do not know if possibly another choice
of red/green bidders and their distributions can give an
unbounded ratio for all non-matroid system.

A broader question our work leaves unanswered is whether it's
possible to design other mechanisms (potentially
non-truthful) that, in the presence of
red and green bidders in non-matroid downward-closed market, can always
guarantee a {\em constant} approximation to Myerson's
revenue on the green bidders alone. For instance, in
\Cref{ex:nonmatroid}, one could possibly consider randomized
mechanisms that ignore a random bidder in the set $\{a,
b\}$ before running Myerson's auction.

\bibliographystyle{alpha}
\bibliography{bibliography}

\clearpage

\appendix

\section{Missing proofs}

Here we provide a proof of \Cref{prop:matroid-update} for
updating the optimal solution of a weighted matroid:

\begin{proof}
Consider running the {\sc Greedy} algorithm in parallel on
both $\Mat_{\mid E - x}$ and $\Mat$ and call these
executions $\mathcal{E}^-, \mathcal{E}$ respectively.

In case $(I + x) \in \Ind$, the downward-closed property of
$\Mat$ guarantees that both executions will make identical
decisions on elements other than $x$ and element $x$ will be
included in the optimal solution of $\Mat$, hence $I^* = I + x$.

For the other case, suppose first that $x = y$, i.e. $x$ is
the min-weight element on $C$. At the time $x$ is inspected,
all other elements of $C$ have already been inspected and
added to the current solution, hence $x$ is not included
since it would violate independence. Therefore, both
executions proceed making identical decisions in every step
and arrive to the same solution $I^* = (I + x) - x = I$.

Now suppose that $x \neq y$. At the time element $x$ is
considered in $\mathcal{E}$, it can be safely included in the solution. The
reason is that if adding $x$ resulted in a circuit $C'$, then
$C' \subsetneq C$ violating the minimality of $C$. The next
step at which the two executions will diverge again is when
considering $y$ --- if they diverged at a previous step it
would again mean that $x$ is part of a circuit $C' \subsetneq
C$ --- at this point $\mathcal{E}$ ignores $y$.

Finally, suppose that the two executions diverge at a later
step on an element $e$ with $w(e) < w(y)$. Denote by $J$ the
current solution $\mathcal{E}^-$ is maintaining and thus $(J
+ x) - y$ is the current solution of $\mathcal{E}$.
There are two reasons the executions might diverge:

\begin{itemize}
\item $(J + e) \in \Ind$ but $( (J + x) - y ) + e \notin
\Ind$.

In this case, there must exist circuit
$C' \subseteq ( (J + x) - y ) + e$ such that $x, e
\in C'$ and $y \notin C'$. Therefore, by \Cref{prop:circuits}
there exists circuit $C''$ such that $e \in C'' \subseteq
(C' \cup C) - x$. This is a contradiction because $C''$ is a
circuit of $J + e$ which was assumed to be independent.

\item $(J + e) \notin \Ind$ but $( (J + x) - y ) + e \in
\Ind$.

This case is similar and the proof is omitted.
\end{itemize}

\end{proof}

\section{A note on revenue monotonicity of VCG}

A revenue monotonicity result  similar to ours is proven for
VCG in matroid markets in \cite{drs09}. We noticed that one
of the propositions used in the proof of that theorem is
incorrect. Here we provide a counter-example and offer an
alternative proof using our \Cref{lem:nonmatroid}.

Proposition 2.9 of \cite{drs09} claims that {\em ``A
downward-closed set system $(U, \Ind)$ with $\Ind \neq
\emptyset$ is a matroid if and only if for every pair $A,B$
of maximal sets in $\Ind$ and $y \in B$, there is some $x
\in A$ such that $A\backslash\{x\} \cup \{y\} \in \Ind$''}.
Here we notice that the ``backward'' direction of this
proposition does not hold.

Consider the following counter-example.
Let $U = \{ a, b, c, d, e \}$ and define $\Ind_3$ to be the
independent sets of the uniform rank $3$ matroid on the
4-element subset $\{ a, b, c, d \}$. Now let $\Ind$ be the
downwards-closed closure of $\Ind_3 \cup \{ \{ a, c, e\},
\{b, d, e\} \}$. We claim $\mathcal{S} = (U, \Ind)$ violates the above
proposition.

\begin{itemize}
\item $(U, \Ind)$ is {\em not} a matroid: This is easy to
see as $I = \{ a, c, e \}, J = \{d, e\}$ violate the
exchange property.

\item Nevertheless, for every pair $A, B$ maximal sets in
$\Ind$ and $y \in B$, there is some $x \in A$ such that $A
\backslash  \{x\} \cup \{y\} \in \Ind$.

First, notice that is suffices to show this for all $y \in B
\backslash A$. Otherwise, if $y \in A \cap B$ then set $x = y$
in which case $A \backslash \{x\} \cup \{y\} = A \in \Ind$.

A second observation is that if both $A, B$ are contained in
$\{a, b, c, d\}$ then the statements holds; after all the
forward direction of the proposition holds and the
restriction of $\mathcal{S}$ on that 4-element subset {\em
is} a matroid.

Finally, notice that $\mathcal{S}$ is symmetric under a
permutation that swaps the roles of $a \leftrightarrow b$
and $c \leftrightarrow d$.

The following table summarizes a case analysis on the choice
of $A, B, y$ and provides a choice of $x$ for each that satisfy
the aforementioned property. The cases that are missing are
equivalent to one of the cases in the table under symmetry.

\begin{equation*}
\begin{array}{|c|c|c||c|}
\hline
A & B & y & x\\
\hline
\{a, c, e\} & \{b, d, e\} & b & e\\
\hline
\{a, c, e\} & \{a, c, d\} & d & e\\
\hline
\{a, c, e\} & \{b, c, d\} & b & e\\
\hline
\{a, c, d\} & \{a, c, e\} & e & d\\
\hline
\{b, c, d\} & \{a, c, e\} & a & b\\
\hline
\{b, c, d\} & \{a, c, e\} & e & c\\
\hline
\end{array}
\end{equation*}

\end{itemize}

We now turn into providing a proof of the ``only if''
direction of Theorem 4.1 in \cite{drs09}.

We use the notation $\ones[S]$ for any subset $S \subseteq
U$ of bidder to mean the bidding profile where every bidder
in $S$ bids $1$ and every bidder in $U \backslash V$ bids
$0$.

\begin{theorem}[{``only if'' direction of \cite[Theorem
4.1]{drs09} }]
\label{thm:corrected}
Let $(U, \Ind)$ be a downward-closed set system that is not
a matoid, then there exists a set $V \subseteq U$ and an
element $x \in V$ such that the revenue of VCG on bid
profile $\ones[V-x]$ exceed the revenue of VCG on bid
profile $\ones[V]$.
\end{theorem}
\begin{proof}

By \Cref{lem:nonmatroid}, there exist $I, J \in \Ind$ with
properties
(\ref{lem:nonmatroid:propK})-(\ref{lem:nonmatroid:propMaximum}).
Let $V = I \cup J$ and let $x$ be an arbitrary element of $I
\backslash J$. We will prove that the revenue of VCG on
$\ones[V]$ is less than the revenue of VCG on $\ones[V -
x]$.

Consider the VCG payment of every bidder
in the bid profile $\ones[V]$. Since $I$ is a
maximum cardinality element of $\Ind|_V$
(\Cref{lem:nonmatroid}, property
(\ref{lem:nonmatroid:propMaximum})),
we may choose $I$ as the set of winners. Let
$W$ denote the intersection of all elements
of $\Ind|_V$ that have cardinality $|I|$.
By property (\ref{lem:nonmatroid:propK}) of
\Cref{lem:nonmatroid}, $I-J \subseteq W$.
Every element of $W$ pays zero, because
for $y \in W$ the maximum cardinality
elements of $\Ind|_{V-y}$ have size $|I|-1$,
hence $y$ could bid zero and still belong to
a winning set. On the other hand, every element
$y \in I-W$ pays 1, because by the definition
of $W$ there is a set $K \in \Ind|_V$ such
that $|K|=|I|$ but $y \not\in K$. If $y$
lowers its bid below 1, then $K$ rather
than $I$ would be selected as the set of
winners, hence $y$ must pay 1 in the VCG
mechanism. Finally, bidders not in $I$ pay
zero because they are not winners. The VCG
revenue on bid profile $\ones[V]$ is therefore
$|I\backslash W|$.

Now recall that $x$ denotes an arbitrary element
of $I\backslash J$, and consider the VCG payment of every
bidder in the bid profile $\ones[V-x]$. Since
$V-x$ is a proper subset of $V$,
$(V-x,\Ind|_{V-x})$ is a matroid.
The rank of this matroid is $|I|-1$,
since \Cref{lem:nonmatroid}, property
(\ref{lem:nonmatroid:propK}) implies
that $\Ind|_{V-x}$ contains no sets
of size $|I|$. We may assume that $I-x$
is chosen as the set of winners of VCG.
Let $J'$ denote a superset of $J$ that
is a basis of $(V-x,\Ind|_{V-x})$.
If $y$ is an element of $(I\backslash J')-x$, the set
$I-x-y$ has strictly fewer elements
than $J'$ so the exchange
axiom implies there is some $z \in J'$
such that $I-x-y+z \in \Ind$.
This set $I-x-y+z$ is a basis of
$(V-x,\Ind|_{V-x})$ that does not
contain $y$, hence the VCG payment of any
$y \in (I\backslash J')-x$ is $1$. Now consider
any $y \in I\backslash W$. By the definition
of $W$, there is a set
$K \in \Ind|_V$ such
that $|K|=|I|$ but $y \not\in K$.
Then $K-x$ is a basis of
$(V-x,\Ind|_{V-x})$ but
$y \not\in K$, implying that
$y$'s VCG payment is 1.

We have shown that in the bid profile
$\ones[V]$, the bidders in $I\backslash W$ pay $1$
and all other bidders pay zero, whereas
in the bid profile $\ones[V-x]$, the
bidders in $I\backslash W$ still pay $1$ and, in
addition, the bidders in $(I\backslash J')-x$ pay $1$.
Furthermore the set $(I\backslash J')-x$ is non-empty.
To see this, observe that $|J'| = |I|-1 = |I-x|$,
but $J' \neq I-x$
because then $J$ would be a subset of $I$,
contrary to our assumption that $I,J$
satisfy \Cref{lem:nonmatroid}, property
(\ref{lem:nonmatroid:propNonEmpty}). Hence,
$I-x$ contains at least one element
that does not belong to $J'$, meaning
$(I \backslash J') -x$ is nonempty. We have thus
proven that the VCG revenue
of $\ones[V-x]$ exceeds the VCG revenue of
$\ones[V]$ by at least $|I-J'-x|$, which is
at least $1$.
\end{proof}

\end{document}